\documentclass[a4paper,11pt]{amsart}
\usepackage{amsmath,amssymb,stmaryrd}
\usepackage{amsfonts}
\usepackage{eucal}
\usepackage{amsthm}
\numberwithin{equation}{section}

\newcommand{\bigpare}[1]{\bigl(#1\bigr)}
\newcommand{\biggpare}[1]{\biggl(#1\biggr)}
\newcommand{\Bigpare}[1]{\Bigl(#1\Bigr)}

\newcommand{\bigbrac}[1]{\bigl[#1\bigr]}

\newcommand{\biggbrac}[1]{\biggl[#1\biggr]}

\newcommand{\normalset}[2]{\{#1\mid#2\}}
\newcommand{\bigset}[2]{\bigl\{#1\bigm|#2\bigr\}}

\newcommand{\norm}[1]{\| #1 \|}
\newcommand{\bignorm}[1]{\bigl\| #1 \bigr\|}

\newcommand{\bigabs}[1]{\bigl| #1 \bigr|}

\newcommand{\jap}[1]{\langle #1 \rangle}

\def\a{\alpha}
\def\b{\beta}

\def\d{\delta}
\def\e{\varepsilon}
\def\f{\varphi}
\def\g{\psi}

\def\i{\mbox{\raisebox{.5ex}{$\chi$}}}
\def\k{\kappa}
\def\l{\lambda}
\def\m{\mu}
\def\n{\nu}
\def\o{\omega}
\def\s{\sigma}
\def\t{\tau}
\def\x{\xi}
\def\y{\eta}
\def\z{\zeta}
\def\th{\theta}
\newcommand{\C}{\Gamma}

\newcommand{\G}{\Psi}

\renewcommand{\O}{\Omega}

\newcommand{\Op}{\mathrm{Op}}

\def\re{\mathbb{R}}
\def\co{\mathbb{C}}
\def\ze{\mathbb{Z}}

\def\pa{\partial}

\newcommand{\supp}{\mathrm{{supp}}}

\DeclareMathOperator*{\slim}{s-lim}

\newtheorem{thm}{Theorem}[section]
\newtheorem{lem}[thm]{Lemma}

\theoremstyle{definition}

\newtheorem{ass}{Assumption}

\theoremstyle{remark}
\newtheorem{rem}{Remark}[section]

\numberwithin{equation}{section}

\title[Scattering matrices for critically long-range perturbations]%
{Remarks on scattering matrices for Schr\"odinger operators with critically 
long-range perturbations}
\author{Shu Nakamura}
\address{
Department of Mathematics, Gakushuin University, 
1-5-1, Mejiro, Toshima, Tokyo, Japan 171-8588}
\email{shu.nakamura@gakushuin.ac.jp}
\thanks{The work is partially supported by JSPS Grant Kiban-B 15H03622. 
The work is inspired by discussions with Dimitri Yafaev during the author's staying at Isaac Newton Institute for Mathematical Sciences for the program: Periodic and Ergodic Spectral Problems, supported by EPSRC Grant Number EP/K032208/1. 
The author thanks Professor Yafaev for the valuable discussion, and the institute and the Simons Foundation for the financial support and its hospitality.
He also thanks Koichi Taira for finding errors in the first version of the paper. } 
\date{\today}

\begin{document}
\maketitle

\begin{abstract}
We consider scattering matrix for Schr\"odinger-type operators on $\re^d$ with perturbation 
$V(x)=O(\jap{x}^{-1})$ as $|x|\to\infty$. We show that the scattering matrix (with time-independent 
modifiers) is a pseudodifferential operator, and analyze its spectrum. 
We present examples of which the spectrum of the scattering matrices have dense point spectrum, 
and absolutely continuous spectrum, respectively. 
\end{abstract}

\section{Introduction}
\label{section-introduction}

In this note, we consider the scattering matrices for Schr\"odinger-type operators
\[
H= H_0 + V\quad\text{on } \mathcal{H}=L^2(\re^d),
\]
where $H_0=p_0(D_x)$ is a Fourier multiplier, and $V=V^W(x,D_x)$ is a long-range 
perturbation of $H_0$. We will explain the general setup in the next section, and here 
we present our main results for the standard Schr\"odinger operators with potential perturbations, i.e., 
$H_0=-\frac12\triangle$, and $V=V(x)$. 
We say the potential $V(x)$ is a long-range perturbation, if $V(x)$ is a real-valued 
smooth function, and there is $\m\in (0,1]$ such that for any multi-index $\a\in \ze_+^d$, 
\[
\bigabs{\pa_x^\a V(x)}\leq C_\a \jap{x}^{-\m-|\a|}, \quad x\in\re^d,
\]
with some $C_\a>0$, where $\jap{x}=(1+|x|^2)^{1/2}$. 
We consider the case $\m\in (0,1)$ in another paper \cite{N2018}, 
and we concentrate on the case $\m=1$ in this paper. Namely, we suppose 

\begin{ass}\label{ass-potential}
$V(x)\in C^\infty(\re^d;\re)$, and for any $\a\in\ze^d$, there is $C_\a>0$ such that 
\[
\bigabs{\pa_x^\s V(x)}\leq C_\a\jap{x}^{-1-|\a|}, \quad x\in\re^d.
\]
\end{ass}

At first, we show the scattering matrix is a pseudodifferential operator, and compute the 
principal symbol. 

\begin{thm}\label{thm-Smatrix-representation}
Under Assumption~A, for any $\l>0$, the scattering matrix $S(\l)\in \mathcal{B}(L^2(S^{d-1}))$ 
is a pseudodifferential operator on $S^{d-1}$, and the principal symbol is given by 
\[
s_0(\l,x,\x) =\exp\biggpare{-i\int_{-\infty}^\infty (V(x+t\sqrt{2\l} \x)-V(t\sqrt{2\l}\x))dt}, 
\]
for $\x\in S^{d-1}$, $x\in T^*_\x S^{d-1}\simeq \x^\perp$. 
More precisely, if we write the symbol of $S(\l)$ by $s(\l,x,\x)$, then 
$s(\l,\cdot,\cdot)\in S^{\d}_{1,0}(T^*S^{d-1})$, and 
$s(\l,\cdot,\cdot)-s_0(\l,\cdot,\cdot)\in S^{-1+\d}_{1,0}(T^*S^{d-1})$ 
with any $\d>0$. 
\end{thm}

\begin{rem}
This is essentially a refined version of a result by Yafaev \cite{Yafaev-1998} for the case $\m=1$, and our proof 
for generalized model follows the argument of  Nakamura \cite{N2016} for short range perturbations. 
This argument works for $\m>1/2$, as in the paper \cite{Yafaev-1998} , though we have more 
precise results if we employ Fourier integral operator formulation as in \cite{N2018}, unless 
$\m=1$. Thus one of the purpose of this note is to fill a gap left in \cite{N2018}. 
\end{rem}

\begin{rem}
By a simple change of integration variable, we have 
\[
s_0(\l,x,\x) =\exp\biggpare{-i(2\l)^{-1/2}\int_{-\infty}^\infty (V(x+t\x)-V(t\x))dt}, 
\]
though the expression in Theorem~\ref{thm-Smatrix-representation} might be more natural since $\sqrt{2\l}\x$ 
is the velocity corresponding to $\x\in S^{d-1}$ at the energy $\l$. 
If we write 
\[
\g(x,\x)=\int_{-\infty}^\infty (V(x+t\x)-V(t\x))dt, \quad \x\in S^{d-1}, x\in T^*_\x S^{d-1}\simeq \x^\perp, 
\]
then it is easy to see that $\g$ satisfies
\[
\bigabs{\pa_x^\a\pa_\x^\b \g(x,\x)}\leq 
\begin{cases} C_{\a\b}\jap{\log\jap{x}},\quad &\text{if }\a=0,\\
C_{\a\b}\jap{x}^{-|\a|}, \quad &\text{if }\a\neq 0, \end{cases}
\]
 for any $\a,\b\in \ze_+^{d-1}$ in a local coordinate. 
 Thus we learn 
\[
 s_0(\l,x,\x) =\exp(-i(2\l)^{-1/2}\g(x,\x))\in S^\d_{1,0}(T^*S^{d-1})
\]
with any $\d>0$. 
\end{rem}

Next, we consider the spectral properties of $S(\l)$ using the above representation. 

\begin{thm}\label{thm-discrete-spectrum}
Suppose Assumption~A, and suppose $V$ is rotation symmetric and
\[
|x\cdot\pa_x V(x)|\geq c|x|^{-1}\quad \text{for }|x|\geq R,
\]
with some $c,R>0$. Then for any $\l>0$ the scattering matrix has dense pure point 
spectrum on the whole unit circle.
\end{thm}

For the moment, we need the rotation symmetry to show the pure pint spectrum, but we can 
show the absence of absolutely continuous spectrum under weaker assumptions. 
We discuss these in Section~\ref{section-dspec}. 

\begin{thm}\label{thm-ac-spectrum}
Suppose $d=2$, and let
\[
V(x)= a \frac{x_1}{\jap{x}^2}, \quad x=(x_1,x_2)\in\re^2,
\]
with $a\neq 0$. Then, $\s_{\mathrm{ess}}(S(\l))=\{e^{i\th}\,|\,|\th|\leq |a|\pi(2\l)^{-1/2}\}$, and 
$S(\l)$ has absolutely continuous spectrum
on $S^1\setminus\{e^{\pm ia\pi(2\l)^{-1/2}}\}$, except for possibly discrete eigenvalues. The eigenvalues 
may accumulate only at $e^{\pm i\pi a(2\l)^{-1/2}}$. 
\end{thm}

The absolutely continuous spectrum is relatively stable under small perturbations, and we have 
the same properties if we add lower order perturbations. 

There is extensive literature concerning the two-body long-range scattering. 
We refer textbooks, Reed-Simon Volume~3 \cite{Reed-Simon} \S X1-9, Yafaev \cite{Yafaev-SLNM1735} Part~2, 
\cite{Yafaev-AMS2009} Chapter~10, Derezi\'nski-G\'erard \cite{Derezinski-Gerard}, and references therein. 
About the scattering matrix for long-range scattering, there are detailed analysis by 
Yafaev, especially \cite{Yafaev-1998}. Our approach is closely related to his result, though 
our formulation is more general and the proof is substantially different. Actually it is 
a direct extension of a previous paper by the author \cite{N2016}. 
In particular, this argument is easily generalized to discrete Schr\"odinger operators 
with long-range perturbations (\cite{N2014}, \cite{Tadano}). 
Our example of scattering matrix with pure point spectrum is discussed in \S 9.7 in 
Yafaev \cite{Yafaev-SLNM1735}, though in a different manner, and we also discuss generalizations. 
Thus the author feel it would be useful to include an independent proof. 
 
Theorems~\ref{thm-discrete-spectrum} and \ref{thm-ac-spectrum} are proved in 
Section~\ref{section-dspec}, and Section~\ref{section-abspec}, respectively. 
In Sections~\ref{section-abspec} we use functional calculus of unitary pseudodifferential operators, 
and for the completeness we give a proof the functional calculus in Appendix~\ref{appendix-funccal}. 
A construction of approximate logarithm of unitary pseudodifferential operators is discussed 
in Appendix~\ref{appendix-log}, and a simple result of trace-class scattering theory for unitary operators  
is discussed in Appendix~\ref{appendix-KB}.

In the following, we use the Weyl quantization of a symbol $a\in C^\infty(\re^{2d})$: 
\[
\Op(a)\f(x)= (2\pi)^{-d}\iint e^{i(x-y)\cdot\x} a(\tfrac{x+y}{2},\x) \f(y)dyd\x, 
\quad \f\in\mathcal{S}(\re^d).
\]
We denote the Kohn-Nirenberg symbol class {\em in $\x$-space}\/ by $S^m_{\rho,\d}$, i.e., 
$a\in S^m_{\rho,\d}$ if $a\in C^\infty(\re^{2d})$ and for any $\a,\b\in \ze_+^d$ there is 
$C_{\a\b}$ such that 
\[
\bigabs{\pa_x^\a\pa_\x^\b a(x,\x)}\leq C_{\a\b} \jap{x}^{m-\rho|\a|+\d|\b|}, 
\quad x,\x\in\re^d. 
\]
We also use the H\"ormander $S(m,g)$ symbol class notation \cite{Hormander}, but we will 
use it for specific metrics $g$ and $\tilde g$, and we explain later. For a symbol class $\Sigma$, 
we denote the corresponding operator set by $\Op \Sigma= \bigset{\Op(a)}{a\in \Sigma}$. 
We refer H\"ormander \cite{Hormander}, Dimassi-Sj\"ostrand \cite{Dimassi-Sjostrand} and 
Zworski \cite{Zworski} for the pseudodifferential operator 
calculus.


\section{Representation formula of the scattering matrix}\label{section-Srep}

Here we define long-range wave operators and scattering operators using time-independent 
modifiers originally dues to Isozaki and Kitada \cite{Isozaki-Kitada1, Isozaki-Kitada2}. We follow the formulation of Nakamura \cite{N2016}, 
and sketch the proof of Theorem~\ref{thm-Smatrix-representation} in a generalized setting. 

\begin{ass}\label{ass-B}
Let $p_0(\x)\in C^\infty(\re^d;\re)$ and elliptic in the following sense: There is $\n>0$ such that 
$p_0\in S^{\n}$, i.e., $\pa_\x^\a p_0(\x) =O(\jap{\x}^{\n-|\a|})$ for any $\a\in\ze_+^d$, and 
\[
p_0(\x)\geq c_0\jap{\x}^\n -c_1, \quad \x\in\re^d, 
\]
with some $c_0,c_1>0$. 
Let $I\Subset\re$ be a compact interval. We suppose there is $c_0>0$ such that 
\[
\bigabs{\pa_\x p_0(\x)}\geq c_0\quad \text{for } \x\in p_0^{-1}(I).
\]
\end{ass}

We set
\[
H_0=p_0(D_x) =\mathcal{F}^* p_0(\cdot)\mathcal{F},
\]
where $\mathcal{F}$ is the Fourier transform, and we also write the free velocity by 
\[
v(\x)= \pa_\x p_0(\x), \quad \x\in\re^d. 
\]

We suppose the perturbation $V$ is a symmetric pseudodifferential operator 
with the real-valued Weyl symbol $V(x,\x)$, i.e., 
\[
V\f (x) = (2\pi)^{-d} \iint e^{i(x-y)\cdot\x} V(\tfrac{x+y}{2},\x) f(y)dyd\x, 
\quad \f\in \mathcal{S}(\re^d). 
\]
We denote the metric $g=dx^2/\jap{x}^2+d\x^2$, and the symbol class 
$S(m,g)$ is defined as follows: $a\in S(m,g)$ if and only if $a\in C^\infty(\re^{2d})$ and 
\[
\bigabs{\pa_x^\a\pa_\x^\b a(x,\x)}\leq C_{\a\b}m(x,\x) \jap{x}^{-|\a|}, 
\quad x,\x\in\re^d
\]
for any $\a,\b\in\ze_+^d$, with some $C_{\a\b}>0$. 

\begin{ass} \label{ass-C}
$V(x,\x)$ is real valued and $V\in S(\jap{x}^{-1}\jap{\x}^\n,g)$. 
\end{ass}

We write
\[
H= H_0 +V =p_0(D_x)+V^W(x,D_x)
\]
be our Hamiltonian, and we suppose:

\begin{ass}\label{ass-D}
$H$ is essentially self-adjoint on $H^\n(\re^d)$. 
\end{ass}

We write the symbol of $H$ by 
\[
p(x,\x)= p_0(\x)+V(x,\x).
\]

\begin{rem}
It might be natural to assume the ellipticity: 
\[
|p(x,\x)|\geq c_0\jap{\x}^\n -c_1, \quad \text{for } x,\x\in\re^d. 
\]
It implies the self-adjointness on $H^\n(\re^d)$, but it is not essential in the following argument. 
\end{rem}

For $\e>0$, we denote 
\[
\O_\pm^\e= \bigset{(x,\x)\in\re^{2d}}{\pm\cos(x,v(\x))>-1+\e, |x|\geq 1, p_0(\x)\in I}.
\]
As well as in \cite{N2016} Section~3, we can construct symbols $a^\pm\in S(1,g)$ such that 
\[
H\Op(a^\pm) -\Op(a^\pm)H_0\sim 0
\]
in the formal symbol sense as $|x|\to\infty$ in $\O_\pm^\e$. 
$a_\pm$ have the form: 
\[
a^\pm(x,\x)\sim e^{i\g_\pm(x,\x)}\bigpare{1+a_1^\pm(x,\x)+a_2^\pm(x,\x) + \cdots}
\]
where 
\[
\g_\pm(x,\x)=\int_0^{\pm\infty} (V(x+tv(\x),\x)-V(tv(\x),\x)) dt.
\]
We note $\g_\pm(x,\x)\notin S(1,g)$ (on $\O_\pm^\e$) in general, but 
for any $\a,\b\in\ze_+^d$, 
\[
\bigabs{\pa_\x^\b \g_\pm(x,\x)}\leq C_\b\jap{\log\jap{x}},
\]
and if $\a\neq 0$, 
\[
\bigabs{\pa_x^\a \pa_\x^\b \g_\pm(x,\x)}\leq C_{\a\b}\jap{x}^{-|\a|}
\]
on $\O_\pm^\e$. We note $\g_\pm$ satisfies 
\[
v(\x)\cdot\pa_x\g_\pm(x,\x) +V(x,\x)=0
\]
as well as in the short-range case (see \cite{N2016} Section~3). 

We introduce a new metric $\tilde g$ by 
\[
\tilde g = \jap{x}^{-2}dx^2 + \jap{\log\jap{x}}^2 d\x^2
\quad\text{on } \re^{2d}.
\]
Then the corresponding symbol class $S(m,\tilde g)$ is defined as follows: 
$a\in S(m,\tilde g)$ if and only if, for any $\a,\b\in\ze_+^d$, 
\[
\bigabs{\pa_x^\a \pa_\x^\b a(x,\x)}\leq C_{\a\b}m(x,\x) \jap{x}^{-|\a|}\jap{\log\jap{x}}^{|\b|}
\]
with some $C_{\a\b}>0$.  We note, hence, for any $\d>0$, $S(m,\tilde g)\subset S(m\jap{x}^\d, g)$. 

By the same construction of $a_j^\pm$ as in \cite{N2016}, Section~3, and direct computations, 
we can easily show $a_j^\pm \in S(\jap{x}^{-j}\jap{\log\jap{x}}^j,\tilde g)$
on $\O_\pm^\e$. 
Hence, $a^\pm$, which is an asymptotic sum of $\{a_j^\pm\}$, is an element of 
$S(1,\tilde g)\subset S(\jap{x}^\d,g)$, with any $\d>0$ on $\O_\pm^\e$. We also note 
$a_\pm-1\in S(\jap{x}^{-1}\jap{\log\jap{x}},\tilde g)\subset S(\jap{x}^{-1+\d},g)$ on $\O_\pm^\e$. 

We choose smooth cut-off functions $\i$, $\z$ and $\y$ such that:  $\i\in C_0^\infty(I)$ with 
$\i(\l)=1$ on $I'\Subset I$; $\z(x)=0$ in a neighborhood of $0$ and $\supp[1-\z]\subset\{|x|\leq 2\}$; 
and $\y(\s)=1$ if $\s>-1+2\e$ and $\y(\s)=0$ if $\s\leq -1+\e$ with sufficiently small $\e>0$. 
With these  cut-off functions, we set
\[
\tilde a^\pm(x,\x)=\i(p_0(\x))\z(|x|)\y(\pm\cos(x,v(\x))) a^\pm(x,\x). 
\]
Then we have symbols $\tilde a^\pm\in S(1,\tilde g)$. 
We set 
\[
J_\pm =\Op(\tilde a^\pm).
\]
We note the principal symbols of $J_\pm^* J_\pm$ are 
$\bigabs{\i(p_0(\x))\z(|x|)\y(\pm\cos(x,v(\x))}^2$, and the remainder terms are in $S(\jap{x}^{-1+\d},g)$. 
Hence $J_\pm$ are bounded in $L^2$, and we can utilize standard pseudodifferential operator 
calculus as if they are in $S(1,g)$. We call $J_\pm$ the {\em time-independent modifiers}, 
or the {\em Isozaki-Kitada modifiers} \cite{Isozaki-Kitada1, Isozaki-Kitada2}. By the construction, 
\[
\mathrm{EssSupp}[a^\pm]\subset
\{p_0(\x)\in I\setminus I'\}\cup\{\pm\cos(x,v(\x))\in [-1+\e,-1+2\e]\}
\cup\{|x|\leq 2\},
\]
where $\mathrm{EssSupp}[\cdot]$ denotes the essential support of the symbol. 
Using this fact and the standard non-stationary phase argument, we can show the existence of modified wave operators: 
\[
W_\pm E_{I'}(H_0) = \slim_{t\to\pm\infty} e^{itH} J_\pm e^{-itH_0}E_{I'}(H_0)
\]
where $E_I(A)$ denotes the spectral projection. We recall $W_\pm$ has the intertwining property: 
\[
H W_\pm E_{I'}(H_0) = W_\pm E_{I'}(H_0) H_0.
\]
We set the (modified) scattering operator $S$ by 
\[
S E_{I'}(H_0) = (W_+E_{I'})^* W_- E_{I'}(H_0),
\]
and then $SE_{I'}(H_0)$ is a unitary operator on $E_{I'}(H_0)\mathcal{H}$. 
By the above intertwining property, $S$ commutes with $H_0$. 

We now define the scattering matrix $S(\l)$ for $\l\in I'$. 
We denote the energy surface with the energy $\l\in I$ by 
\[
\Sigma_\l = \bigset{\x\in\re^d}{p_0(\x)=\l}=p_0^{-1}(\{\l\}).
\]
We note $\Sigma_\l$ is a smooth hypersurface by the above assumption. 
Let 
\[
m_\l =|p_0(\x)|^{-1} dS(\x)
\]
be a measure on $\Sigma_\l$, where $dS(\x)$ is the surface measure on $\Sigma_\l$, so that 
\[
\int \f d\x = \int_I \biggpare{\int_{\Sigma_\l} \f\big|_{\Sigma_\l} dm_\l}d\l 
\]
for $\f\in C_0^\infty(p_0^{-1}(I))$. Hence we have the integral decomposition 
\[
L^2(p_0^{-1}(I),d\x) \simeq \int^\oplus_I L^2(\Sigma_\l,m_\l) d\l.
\]
Since $S$ commutes with $H_0$, the operator $\mathcal{F} SE_{I'}(H_0)\mathcal{F}^*$ 
commutes with $p_0(\x)\cdot$, and hence it is decomposed to operators on $L^2(\Sigma_\l,m_\l)$: 
\[
\mathcal{F} SE_{I'}(H_0)\mathcal{F}^* \simeq \int^\oplus_{I'} S(\l) d\l 
\quad\text{on } \int^\oplus_{I'} L^2(\Sigma_\l,m_\l) d\l .
\]
The family of operators $\{S(\l)\}_{\l\in I'}$ is called the scattering matrix. 

Given the above construction, we can prove the following theorem in exactly the same 
argument as in \cite{N2016} (see also \cite{N2018}). We note the microlocal resolvent estimate, which is crucial 
in the proof, is proved in \cite{N2017} under our setting. 

\begin{thm}\label{thm-Srep}
Let $\l\in I'\setminus\s_{\mathrm{p}}(H)$. Then $S(\l)$ is a pseudodifferential operator on $\Sigma_\l$. 
If we denote the symbol by $s(\l,x,\x)$, then it satisfies for any $\a,\b\in\ze_+^{d-1}$, 
\[
\bigabs{\pa_x^\a \pa_\x^\b s(\l,x,\x)}\leq C_{\a\b} \jap{x}^{-|\a|}\jap{\log\jap{x}}^{|\b|}
\]
for $\x\in\Sigma_\l$, $x\in T^*_\x\Sigma_\l$. Moreover, the principal symbol is given by 
\[
s_0(\l,x,\x) =\exp\biggpare{-i\int_{-\infty}^\infty (V(x+tv(\x),\x)-V(tv(\x),\x))dt},
\]
i.e., $s(\l,\cdot,\cdot)-s_0(\l,\cdot,\cdot)\in S(\jap{x}^{-1+\d},g)$ with any $\d>0$. 
\end{thm}


\section{Scattering matrix with pure point spectrum}\label{section-dspec}

We first note that, if $H_0=-\frac12\triangle$, 
 and if the perturbation is rotation symmetric, then the scattering matrix 
is also rotation symmetric. Then we can easily show that such operator has pure point spectrum. 
This model is also discussed in \cite{Yafaev-SLNM1735} \S9.7. 

\begin{lem}\label{lem-dspec-1}
Suppose $U$ is a rotation symmetric bounded pseudodifferential operator on $S^{d-1}$, then 
the spectrum is pure point. 
\end{lem}

\begin{proof}
In the geodesic local coordinate with the center at $\x_0$, the symbol of 
the operator $U$ has the form $u(\x_0,|x|^2)$ by virtue of the symmetry 
(with respect the rotation around $\x_0$). Then, again by the symmetry, the 
symbol is independent of $\x_0$, i.e., the symbol has the form $u(\x,|x|^2)=g(|x|^2)$ 
in the geodesic local coordinate. This implies $U= g(-\triangle)$, where 
$\triangle$ is the Laplace-Beltrami operator on $S^{d-1}$. 
Since the spectrum of $-\triangle$ is pure point, the spectrum of $U=g(-\triangle)$ 
is also pure point. 
\end{proof}

We now observe the spectrum of the scattering matrix tends to cover the whole unit circle. 

\begin{lem}\label{lem-sym-spec}
Suppose $V=V(x)$ is a rotationally symmetric potential and satisfies Assumption~A. 
Suppose, moreover, $V$ satisfies 
\[
|x\cdot\pa_x V(x)| \geq c|x|^{-1}, \quad |x|\geq R, 
\]
with some $c>0$ and $R>0$. Then for any $\l>0$,  $\s(S(\l))=S^1= \normalset{z\in\co}{|z|=1}$. 
\end{lem}

\begin{proof}
We suppose $x\cdot\pa_x V(x)\geq c_0|x|^{-1}$ for large $x$. 
Let $\th_0\in [0,2\pi]$ be fixed, and we show $e^{-i\th_0}\in \s(S(\l))$. 
We write $V(x)= g(|x|)$. 

We write, for $\x\in S^{d-1}$, $x\perp \x$ and $|x|\geq R$, 
\begin{align}
\g(x,\x) &= \int_{-\infty}^\infty (V(x+t\x)-V(t\x))dt, \nonumber\quad \\
&=\int_{-\infty}^\infty \biggpare{\int_0^1x\cdot\pa_x V(sx+t\x)ds} dt. \label{eq-psi}
\end{align}
We note, since $V(x)$ is rotationally symmetric, we have 
\[
x\cdot\pa_x V(x) =|x|g'(|x|)\geq c_0|x|^{-1}, 
\]
and hence 
\begin{align*}
x\cdot\pa_x V(sx+t\x) &= x\cdot \frac{sx+t\x}{|sx+t\x|} g'(|sx+t\x|)\\
&= \frac{s|x|^2}{|sx+t\x|}g'(|sx+t\x|)
\geq \frac{c_0 s|x|^2}{\jap{sx+t\x}^3}.
\end{align*}
Thus we have 
\begin{align*}
\g(x,\x) &\geq \int_{-\infty}^\infty \biggpare{\int_0^1 \frac{c_0s|x|^2}{\jap{sx+t\x}^3}ds} dt\\
&=  \int_0^1\biggpare{ \int_{-\infty}^\infty\frac{c_0s|x|^2}
{(s^2|x|^2+t^2+1)^{3/2}}dt} ds \\
&=2c_0\int_0^1 \frac{s|x|^2}{s^2|x|^2+1}ds 
=2c_0\int_{0}^{|x|} \frac{s ds}{s^2+1} 
=2c_0\log\jap{x}.
\end{align*}
Here we have used the formula: $\int_0^\infty (a^2+t^2)^{-3/2}dt =a^{-2}$, $a>0$. 
In particular $\g(x,\x)\to\infty$ as $|x|\to\infty$, and hence, 
for any $N>0$ we can find $(x_N,\x_N)$ such that $|x_N|\geq N$ and 
$\g(x_N,\x_N)\equiv \l\th_0\mod (2\pi\ze)$. 
We set 
\[
\f_N(\x) = c_N \exp(ix_N\cdot(\x-\x_N) - |\x-\x_N|^2/|x_N|)
\]
in a neighborhood inside a local coordinate of $\x_N$, where $c_N$ is chosen so that $\|\f_N\|=1$. 
Then $\f_N$ is supported essentially in 
\[
\normalset{(x,\x)}{|x-x_N|=O(\jap{x_N}^{1/2}), |\x-\x_N|=O(\jap{x_N}^{-1/2}}. 
\]
We also recall $e^{-i(2\l)^{-1/2}\g(x,\x)}$ is the principal symbol of $S(\l)$, and $\pa_x\g(x,\x) = O(|x|^{-1})$, 
$\pa_\x\g(x,\x)= O(\log\jap{x})$ as $|\x|\to\infty$. These imply
\[
\jap{\f_N,S(\l)\f_N}-e^{-i\th_0}\norm{\f_N}^2 =O(\jap{x_N}^{-1/2}\log\jap{x_N}) \to 0 
\quad\text{as }N\to\infty,
\]
and we may assume $\{\f_N\}$ are asymptotically orthogonal (since they have essentially 
disjoint supports in the phase space).
Then by the Weyl's criterion (\cite{Reed-Simon} Theorem~VII.12), 
we conclude $e^{i\th}\in \s_{\mathrm{ess}}(S(\l))$. 
The proof for the case $x\cdot\pa_x V(x)\leq -c_0|x|^{-1}$ ($|x|\geq R$) is essentially the same. 
\end{proof}

Theorem~\ref{thm-discrete-spectrum} follows immediately from the above two lemmas. 

We now consider slightly more general potentials. We write
\[
\pa_r f(x) = \hat x\cdot\pa_x f(x), \quad \hat x=\frac{x}{|x|},
\]
and 
\[
\pa_r^\perp f(x) =\pa_x f(x) -\pa_r f(x)\hat x
=(E-\hat x\otimes \hat x )\pa_x f(x),
\]
for $f\in C^1(\re^d)$. 

\begin{thm}\label{thm-no-ac-spectrum}
Suppose $V$ satisfies Assumption~\ref{ass-potential}, and there are constants $c_1, c_2, R>0$ 
such that $c_1>c_2$ and 
\begin{equation}\label{eq-no-ac-condition}
\bigabs{\pa_r V(x)}\geq \frac{c_1}{|x|^2}, \quad 
\bigabs{\pa_r^\perp V(x)}\leq \frac{c_2}{|x|^2}, \quad \text{if }|x|\geq R.
\end{equation}
Then $\s(S(\l)) =S^1$, and $S(\l)$ has no absolutely continuous spectrum for $\l>0$. 
\end{thm}

\begin{rem}
Suppose  $V(x)=-f(\th)/r$, $x=(r\cos\th,r\sin\th)\in\re^2$ for $|x|\geq R$, $f(\th)>0$. 
Then the condition 
\eqref{eq-no-ac-condition} is equivalent to 
\[
\inf_\th f(\th)=c_1 >c_2= \sup_\th |f'(\th)|.
\]
\end{rem}

\begin{lem}\label{lem-no-ac-condition}
Suppose $V$ satisfies \eqref{eq-no-ac-condition}, then there is $c_3>0$ such that 
\[
\g(x,\x)\geq 2(c_1-c_2)\log|x| -c_3, \quad \x\in S^{d-1}, x\perp \x.
\]
\end{lem}

\begin{proof}
Here we suppose $\pa_r V(x)\geq c_1/|x|^2$. The other case is considered similarly. 
We may suppose $|x|\geq R$ without loss of generality. 
We recall \eqref{eq-psi}. 
We write $y=sx+t\x$, and compute 
\[
x\cdot \pa_x V(y) = \pa_r V(y) (x\cdot\hat y) +x\cdot \pa_r^\perp V(y).
\]
At first, we note
\[
x\cdot \hat y = \frac{x\cdot (sx+t\x)}{|sx+t\x|} = \frac{s|x|^2}{(s^2|x|^2 +t^2)^{1/2}}.
\]
We also note
\begin{align*}
(E-\hat y\otimes \hat y)x 
&= x-(x\cdot\hat y)\hat y 
= x-\frac{s|x|^2(sx+t\x)}{s^2|x|^2 +t^2} \\
&= \frac{(s^2|x|^2+t^2)-s^2|x|^2}{s^2|x|^2+t^2}x -\frac{s|x|^2 t}{s^2|x|^2+t^2}\x\\
&= \frac{t^2 x-st|x|^2\x}{s^2|x|^2+t^2}, 
\end{align*}
and thus 
\[
\bigabs{(E-\hat y\otimes \hat y)x } 
= \frac{(t^4|x|^2 +s^2t^2|x|^4)^{1/2}}{s^2|x|^2+t^2} 
= \frac{|t||x|}{(s^2|x|^2+t^2)^{1/2}}.
\]
Hence we learn 
\begin{align*}
\int_{-\infty}^\infty \pa_r V(y)(x\cdot\hat y)dt 
&\geq \int_{-\infty}^\infty \frac{c_1}{|sx+t\x|^2}\cdot \frac{s|x|^2}{(s^2|x|^2+t^2)^{1/2}} dt \\
&= \int_{-\infty}^\infty \frac{c_1s|x|^2 dt}{(s^2|x|^2+t^2)^{3/2}}
= \frac{2c_1s|x|^2}{s^2|x|^2} =\frac{2c_1}{s}, 
\end{align*}
provided $s|x|\geq R$. 
Similarly, we learn
\begin{align*}
\int_{-\infty}^\infty \bigabs{x\cdot\pa_r^\perp V(y)}dt 
&\leq \int_{-\infty}^\infty \frac{c_2}{|sx+t\x|^2}\cdot \frac{|t||x|}{(s^2|x|^2+t^2)^{1/2}} dt \\
&= \int_{-\infty}^\infty \frac{c_2|x||t| dt}{(s^2|x|^2+t^2)^{3/2}}
= \frac{2c_2|x|}{s|x|} = \frac{2c_2}{s},
\end{align*}
if $s|x|\geq R$.
Here we have used the formula: $\int_0^\infty t(a^2+t^2)^{-3/2}dt = a^{-1}$. 
Thus we have 
\begin{align*}
&\int_{R/|x|}^1  \biggpare{\int_{-\infty}^\infty x\cdot\pa_x V(sx+t\x) dt} ds
\geq \int_{R/|x|}^1 \frac{2(c_1-c_2)}{s} ds \\
&\quad = 2(c_1-c_2)\log(|x|/R)
=2(c_1-c_2)\log|x|-2(c_1-c_2)\log R. 
\end{align*}
On the other hand, if $s|x|\leq R$, we use 
\[
\bigabs{x\cdot \pa_x V(sx+t\x)}\leq C|x|\jap{t\x}^{-2}= C|x|\jap{t}^{-2},
\]
with some $C>0$, which follows directly from Assumption~\ref{ass-potential}. Hence, we learn
\[
\int_0^{R/|x|}  \biggpare{\int_{-\infty}^\infty \bigabs{x\cdot\pa_x V(sx+t\x)} dt} ds
\leq C|x|\cdot\frac{R}{|x|}\int_{-\infty}^{\infty}\jap{t}^{-2}dt 
=C\pi R. 
\]
Combining these, we obtain
\[
\int_0^1  \biggpare{\int_{-\infty}^\infty x\cdot\pa_x V(sx+t\x) dt} ds
\geq 2(c_1-c_2)\log|x| - c_3, 
\]
where $c_3= 2(c_1-c_2)\log R +C\pi R$. 
\end{proof}

\begin{proof}[Proof of Theorem~\ref{thm-no-ac-spectrum}]
The claim $\s(S(\l))=S^1$ is proved exactly as in the proof of Lemma~\ref{lem-sym-spec} 
using Lemma~\ref{lem-no-ac-condition}.

By Theorem~\ref{thm-log-unitary} in Appendix~\ref{appendix-log}, we learn there is a real-valued symbol 
$\G\in S(\jap{\log\jap{x}},g)$ 
such that $S(\l) \equiv \exp(-i(2\l)^{-1/2}\Op(\G))$ modulo $S(\jap{x}^{-\infty},g)$, where 
$g=dx^2/\jap{x}^2+d\x^2$. Moreover, the principal symbol of $\G$ is 
$\g$ computed above, i.e., $\G-\g\in S(\jap{x}^{-1+\d},g)$ with any $\d>0$. 
Then, by Lemma~\ref{lem-no-ac-condition}, $\Op(\G)$ has discrete spectrum, and 
hence $\exp(-i(2\l)^{-1/2}\Op(\G))$ has pure point spectrum. Now we note 
$K= S(\l)-\exp(-i(2\l)^{-1/2}\Op(\G))\in \Op S(\jap{x}^{-\infty},g)$ is a trace class operator, and 
we can apply the scattering theory for trace class perturbation (see Appendix~\ref{appendix-KB}) to conclude 
$\s_{\mathrm{ac}}(S(\l))= \s_{\mathrm{ac}}(\exp(-i(2\l)^{-1/2}\Op(\G))) =\emptyset$.
\end{proof}
\section{Scattering matrix with absolutely continuous spectrum}\label{section-abspec}

Here we suppose $d=2$, and consider the potential 
\[
V(x)= a\frac{x_1}{\jap{x}^2}, \quad x=(x_1,x_2)\in\re^2.
\]

At first we compute the principal part of 
$\g(x,\x)= \int_{-\infty}^\infty (V(x+t\x)-V(t\x))dt$ for 
$ |\x|=1$, $x\perp\x$. We use the standard coordinate for $S^1$: We denote a point
$\x\in S^1$ by $\th\in \mathbb{T}=\re/2\pi\ze$ such that
\[
\x=(\cos\th,\sin\th), \quad \th\in [0,2\pi)\simeq \mathbb{T}.
\]
The cotangent space at $\th$ is identified with the orthogonal space at $\th$, i.e., 
\[
x= (-\o\sin\th, \o \cos\th), \quad \o\in\re. 
\]
We use $(\th,\o)\in \mathbb{T}\times\re$ as the coordinate system of $T^*S^1$. 
As in the last section, we write 
\[
\g(x,\x)=\int_{-\infty}^\infty (V(x+t\x)-V(t\x))dt
\]
so that $\exp(-i(2\l)^{-1/2}\g(x,\x))$ is the principal symbol of $S(\l)$. 

\begin{lem}\label{lem-abspec-1}
Let $V$ and the coordinate of $T^*S^1$ as above. Then 
\[
\g(x,\x)=-a\pi \sin\th\frac{\o}{\jap{\o}}, \quad (\th,\o)\in T^*S^1.
\] 
\end{lem}

\begin{proof}
We again recall \eqref{eq-psi} and we compute 
\[
\pa_x V(x) = \biggpare{\frac{a}{\jap{x}^2},0} +a\biggpare{\frac{-2x_1^2}{\jap{x}^4}, 
\frac{-2x_1x_2}{\jap{x}^4}} 
=\biggpare{\frac{a}{\jap{x}^2},0}-\frac{2ax_1}{\jap{x}^4}\,x.
\]
Then we have 
\begin{align*}
x\cdot \pa_x V(sx+t\x) &= \frac{ax_1}{\jap{sx+t\x}^2}- 2a\frac{sx_1+t\x_1}{\jap{sx+t\x}^4}x\cdot(sx+t\x) \\
&= \frac{ax_1(s^2|x|^2 +t^2+1)-2as^2x_1|x|^2}{(s^2|x|^2+t^2+1)^2}
-\frac{2as|x|^2\x_1 t}{(s^2|x|^2+t^2+1)^2}\\
&= ax_1\frac{t^2-s^2|x|^2+1}{(s^2|x|^2+t^2+1)^2}
-\frac{2as|x|^2\x_1 t}{(s^2|x|^2+t^2+1)^2}.
\end{align*}
Now we note
\[
\int_{-\infty}^\infty \frac{2s|x|^2\x_1 t}{(s^2|x|^2+t^2+1)^2} dt =0
\]
since the integrand is odd. We also note, since 
\[
\frac{d}{dt}\biggpare{\frac{t}{b^2+t^2}} =\frac{b^2-t^2}{(b^2+t^2)^2}, \quad b>0, 
\]
we have 
\[
\int_{-\infty}^\infty \frac{b^2-t^2}{(b^2+t^2)^2}dt =\lim_{T\to\infty} \biggbrac{\frac{t}{b^2+t^2}}_{-T}^T =0.
\]
Using this, we learn 
\begin{align*}
\int_{-\infty}^\infty \frac{t^2-s^2|x|^2+1}{(s^2|x|^2+t^2+1)^2} dt 
&= \int_{-\infty}^\infty \biggpare{\frac{t^2-s^2|x|^2-1}{(s^2|x|^2+t^2+1)^2}
+\frac{2}{(s^2|x|^2+t^2+1)^2}} dt \\
&= \int_{-\infty}^\infty \frac{2}{(s^2|x|^2+t^2+1)^2} dt
= \pi (s^2|x|^2+1)^{-3/2}.
\end{align*}
Here we have used the well-known formula: 
$\int_{-\infty}^\infty (b^2+t^2)^{-2}dt =\pi/(2b^{3})$. Combining these, we learn 
\[
\g(\th,\o)= a\pi \int_0^1\frac{x_1}{\jap{sx}^3}ds 
= a\pi\frac{ x_1}{|x|}\int_0^{|x|} \frac{du}{\jap{u}^3}
=a\pi \frac{x_1}{|x|}\cdot\frac{|x|}{\jap{x}} = a\pi \frac{x_1}{\jap{x}}.
\]
We then substitute $x_1=-\o\sin\th$ and $|x|=|\o|$ to conclude the assertion. 
\end{proof}

Then the essential spectrum of $S(\l)$ is easy to locate using the Weyl theorem: 

\begin{lem}\label{lem-abspec-2}
For the above Hamiltonian, we have 
\[
\s_{\mathrm{ess}}(S(\l))= \bigset{e^{i\t}}{|\t|\leq |a|\pi(2\l)^{-1/2}}, \quad \l>0.
\]
In particular, if $|a|\geq \sqrt{2\l}$ then the essential spectrum is the whole circle. 
\end{lem}

Now we construct a simple scattering theory to show that the essential spectrum is absolutely 
continuous. We set 
\[
q(\th,\o)= \mathrm{sgn}(a)\cos\th \jap{\o}, \quad (\th,\o)\in T^*S^1, 
\]
and we define an operator $Q$ on $L^2(S^1)$ by 
\[
Q=\Op(q) \equiv \mathrm{sgn}(a)\cos\th \jap{-D_\th}\quad \mod \Op(S^0_{1,0}).
\]
We note, since we are working in $\th$-space, it is convenient to quantize function $a(x,\x)$ 
as $a(-D_\th,\th)$. 
We may assume $Q$ is formally self-adjoint, since we may quantize it, 
for example, by 
\[
Qf(\th) =\frac{1}{2\pi}\iint e^{-i(\th-\t)\o}\y(\th-\t)q(\tfrac{\th+\t}{2},\o)f(\t)d\t d\o,
\]
where $\y\in C^\infty(\mathbb{T})$ such that $\y(\t)=1$ if $|\t|\leq 1/8$; $=0$ if $|\t|\geq 1/4$, 
and $f\in C^\infty(\mathbb{T})$, and this $Q$ is formally self-adjoint. 

\begin{lem}\label{lem-abspec-3}
$Q$ is essentially self-adjoint on $H^1(\mathbb{T})$. 
\end{lem}

\begin{proof}
We set $N=\jap{D_\th}$ on $L^2(\mathbb{T})$. Then it is easy to see $N$ is self-adjoint 
with $\mathcal{D}(N)=H^1(\mathbb{T})$ and $N\geq 1$. Moreover, by symbol calculus, 
it is easy to see $Q$ and $[N,Q]$ are  bounded from $H^{1/2}(\mathbb{T})$ 
to $H^{-1/2}(\mathbb{T})$, since the symbols of $Q$ and $[N,Q]$ are in $S^1_{1,0}$. 
Hence, by the commutator theorem (\cite{Reed-Simon} Theorem~X.36), $Q$ is 
essentially self-adjoint on $H^1(\mathbb{T})$. 
\end{proof}

Now we note, $[Q,S(\l)]$, $[Q,[Q,S(\l)]]$, etc., are bounded in $L^2(\mathbb{T})$ 
since symbols of these operators are in $S^0_{1,0}$. Namely, $S(\l)$ is $Q$-smooth 
in the sense of the Mourre theory. 

\begin{lem}\label{lem-abspec-Mourre}
Suppose $I\subset S^1$ be a compact interval such that $I\cap \{e^{\pm ia\pi(2\l)^{-1/2} }\}=\emptyset$. 
Then there is $c>0$ and a compact operator $K(\l)$ such that 
\[
E_I(S(\l)) S(\l)^* [Q, S(\l)] E_I(S(\l)) \geq c E_I(S(\l)) +K(\l),\quad \l>0,  
\]
where $E_I(S)$ denotes the spectral projection for a unitary operator $S$. 
\end{lem}

\begin{proof}
For simplicity, we suppose $a>0$. The other case is similar. 

Let $f\in C_0^\infty(S^1)$. Then using the functional calculus of unitary pseudodifferential 
operators, Theorem~\ref{thm-functional-calculus}, we learn the principal symbol of 
$f(S(\l))S(\l)^*[Q,S(\l)]f(S(\l))$  is given by 
\[
i(f\circ s_0(\l;\cdot))^2 s_0(\l;\cdot)^* \{q,s_0(\l;\cdot)\}
=-(f\circ s_0(\l;\cdot))^2\{q,a\pi(2\l)^{-1/2}\sin\th(\o/\jap{\o})\},
\]
where $\{\cdot,\cdot\}$ denotes the Poisson bracket. By direct computations, we have 
\begin{align*}
-\{\cos\th\jap{\o},\sin\th(\o/\jap{\o})\} 
&= \sin\th \jap{\o}\cdot\sin\th \jap{\o}^{-3}+\cos\th  \o\jap{\o}^{-1}\cdot \cos\th \o\jap{\o}^{-1}\\
&= \frac{\sin^2\th}{\jap{\o}^2} + \cos^2\th \frac{\o^2}{\jap{\o}^2}
\geq \cos^2\th \frac{\o^2}{\jap{\o^2}},
\end{align*}
and hence 
\[
-\{q,a\pi(2\l)^{-1/2}\sin\th(\o/\jap{\o})\}
\geq a\pi(2\l)^{-1/2}  \cos^2\th \frac{\o^2}{\jap{\o}^2}.
\]
Now we choose $I'\Subset S^1$ so that $I\Subset I'$ and $I'\cap\{e^{\pm ia\pi(2\l)^{-1/2} }\}=\emptyset$, 
and then choose $f\in C^\infty(\mathbb{T};\re)$ such that $f=1$ on $I$ and $\supp[f]\subset I'$. Then, by this condition, 
$a\pi(2\l)^{-1/2} \sin\th \neq\pm a\pi (2\l)^{-1/2} $ on the support of $f\circ s_0$, and hence 
$|\sin\th|\leq (1-\e^2)^{1/2}$ with some $\e>0$, i.e., $\cos^2\th\geq \e^2$. 
Thus we learn 
\[
i(f\circ s_0(\l;\cdot))^2 s_0(\l;\cdot)^* \{q,s_0(\l;\cdot)\}\geq 
\e^2 (f\circ s_0(\l;\cdot))^2 \frac{{\o}^2}{\jap{\o}^2}, 
\]
and this implies 
\[
f(S(\l))S(\l)^*[Q,S(\l)]f(S(\l)) \geq \e^2 f(S(\l))^2 +K_1(\l)
\]
with some compact operator $K_1(\l)$ on $L^2(S^1)$. Then, multiplying $E_I(S(\l))$ 
from the both sides, we arrive at the assertion. 
\end{proof}

Then, by the Mourre theory for unitary operators (see, e.g., Fern\'andez-Richard-Tiedra \cite{FRT}), 
we have the following result: 

\begin{thm}\label{thm-abspec}
Let $H$ and $S(\l)$ be as above, and let $\l>0$. Let $\C$ be the set of eigenvalues of $S(\l)$. 
Then $\C$ can accumulate only at $\{e^{\pm ia\pi(2\l)^{-1/2}}\}$. 
For $\x\in S^1\setminus \{\C\cup\{e^{\pm ia\pi\l}\}\}$, the limits
\[
\lim_{\e\downarrow 0} \jap{Q}^{-1}(S(\l)-(1\pm\e)\x)^{-1} \jap{Q}^{-1}
= \jap{Q}^{-1}(S(\l)-(1\pm 0)\x)^{-1}\jap{Q}^{-1}
\]
exist. Hence, in particular, $\s_{\mathrm{sc}}(S(\l))=\emptyset$ and the spectrum of $S(\l)$ 
is absolutely continuous on $S^1\setminus \C$. 
\end{thm}

Theorem~\ref{thm-ac-spectrum} follows immediately from the above theorem. \qed

\appendix
\section{Functional calculus of unitary pseudodifferential operators}\label{appendix-funccal}

In Appendices~\ref{appendix-funccal} and \ref{appendix-log}, 
we consider pseudodifferential operators on $\re^d$, but it can be generalized easily 
to pseudodifferential operators on manifolds. We restrict ourselves to the $\re^d$ case mostly 
to simplify notations related to Beal's characterization of pseudodifferential operators. 

Let $\d\in [0,1)$, and we consider a unitary operator $U$ on $L^2$ with the symbol 
$u\in \bigcap_{\d>0}S^\d_{1,0}$. We consider operators on $\re^d$, or in a local coordinate 
in a $d$-dimensional manifold. 
We show that $f(U)$, the function of $U$, is a pseudodifferential operator,
and compute the principal symbol. At first we note 

\begin{lem}\label{lem-funccal-1}
Suppose $a\in S^1_{1,0}$, and the symbol is bounded. Then $\Op(a)$ is bounded in $L^2$. 
\end{lem}

\begin{proof}
The proof is essentially the same as the G{\aa}rding inequality. 
Without loss of generality, we may suppose $a$ is real valued, and we write $A=\Op(a)$.
Let $M>\sup|a|$. 
We set $b(x,\x)= (M^2-a(x,\x)^2)^{1/2}\in S^1_{1,0}$, and $B=\Op(b)$. Then by the symbol  
calculus, we learn 
\[
R=A^* A+ B^* B -M^2 \in \mathrm{Op}(S^0_{1,0}).
\]
Hence 
\[
\norm{A u}^2 \leq \norm{Au}^2+\norm{Bu}^2 \leq M^2 \norm{u}^2 + \norm{Ru}\norm{u}
\leq C\norm{u}^2
\]
since $R$ is bounded in $L^2$. 
\end{proof}

\begin{lem}\label{lem-funccal-2}
Suppose $U=\Op(u)$ is unitary with $u\in S^\d_{1,0}$, $\d\in [0,1)$. 
Then for any $s\in\re$, 

\[
\bignorm{U^k}_{H^s\to H^s}\leq C_s\jap{k}^{|s|/(1-\d)}, \quad k\in\ze. 
\]
\end{lem}

\begin{proof}
We let $\n=1-\d\in (0,1]$, $s=N\n$, and show 
\[
\bignorm{U^k}_{H^{N\n}\to H^{N\n}} \leq C\jap{k}^N, \quad k\in \ze. 
\]
We first suppose $k>0$. We consider the commutator: 
\[
[\jap{D_x}^\n, U^k] = \sum_{j=1}^{d-1} U^j [\jap{D_x}^\n,U] U^{k-1-j}. 
\]
Since the symbol of the operator $[\jap{D_x}^\n,U]$ is in $S^0_{1,0}$, 
it is bounded in $L^2$, and hence $\bignorm{[\jap{D_x}^\n,U^k]}\leq C\jap{k}$. 
This implies $\norm{U^k}_{H^\n\to H^\n}\leq C\jap{k}$. 

More generally, we compute
\begin{align*}
[\jap{D_x}^{N\n}, U^k] &=\sum_{j=1}^{k-1} U^j[\jap{D_x}^{N\n},U] U^{k-1-j}\\
&= \sum_{j=1}^{k-1} \sum_{\ell=0}^{N-1} U^j \jap{D_x}^{\ell \n}
[\jap{D_x}^\n,U]\jap{D_x}^{(N-1-\ell)\n} U^{k-1-j}.
\end{align*}
Now we use the induction in $N$. Suppose the claim holds for $N\leq N_0$. 
Then we have  
\begin{align*}
&[\jap{D_x}^{N_0\n},U^k] \jap{D_x}^{-N_0\n}\\
&\quad = \sum_{j=1}^{k-1} \sum_{\ell=0}^{N_0-1} U^j \jap{D_x}^{\ell \n}
[\jap{D_x}^\n,U]\jap{D_x}^{(N_0-1-\ell)\n} U^{k-1-j}\jap{D_x}^{-N_0\n}\\
&\quad = \sum_{j=1}^{k-1} \sum_{\ell=0}^{N_0-1} U^j 
(\jap{D_x}^{\ell \n}[\jap{D_x}^\n,U]\jap{D_x}^{-\ell\n}) \times \\
&\qquad\times (\jap{D_x}^{(N_0-1)\n} U^{k-1-j}\jap{D_x}^{-(N_0-1)\n}) \jap{D_x}^{-1}. 
\end{align*}
By the induction hypothesis and the fact $[\jap{D_x}^\n,U]$ is bounded in $H^{\ell\n}$, 
each term in the sum is bounded in $L^2$, and the norm is  
$O(\jap{k}^{(N_0-1)\n})$. By summing up these norms, we arrive at the claim with 
$N=N_0$. 
For $k<0$, we use the same argument for $U^{-1}=U^*$. 
Then the assertion for general $s\in\re$ follows by the interpolation and the duality argument. 
\end{proof}

Now we consider functional calculus of a unitary operator $U$. 
For $f\in C^\infty(S^1)$, we write the Fourier series expansion by $\hat f[k]$, i.e., 
\[
\hat f[k]= \frac{1}{2\pi}\int_0^{2\pi} e^{-ik\th} f(e^{i\th})d\th, \quad k\in\ze, 
\]
and hence 
\[
f(e^{i\th}) =\sum_{k\in\ze} \hat f[k] e^{ik\th}, \quad \th\in [0,2\pi).
\]
We recall $\hat f[n]$ is rapidly decreasing in $n$. Then we write
\[
f(U)=\sum_{k\in\ze} \hat f[k] U^k \in \mathcal{B}(L^2).
\]
It is well-known that $f(U)$ is the same function of $U$ defined in terms of the 
spectral decomposition. We show $f(U)$ is a pseudodifferential operator using the 
Beals characterization of pseudodifferential operators. 

For an operator $A$, we write
\[
K_jA= i[D_{x_j},A],\quad L_j A= -i [x_j,A],
\quad j=1,\dots, d, 
\]
and multiple commutators by $L^\a A$, $K^\b A$, etc., for $\a,\b\in\ze_+^d$. 
We recall $A=\Op(a)$ with $a\in S^\d_{1,0}$ if and only if 
$K^\a L^\b A$ is bounded from $L^2$ to $H^{-\d+|\b|}$ for any $\a,\b\in\ze_+^d$
 (cf. Dimassi-Sj\"ostrand \cite{Dimassi-Sjostrand}, Zworski \cite{Zworski}).  We compute 
\begin{align*}
K^\a L^\b (U^k) = \sum_{\substack{\a^1+\cdots+\a^N =\a, \\
\b^1+\cdots+\b^N=\b,\\ \a^j+\b^j\neq 0, \\k_1+\cdots+k_{N+1}=k}}
 U^{k_1} (K^{\a^1}L^{\b^1} U) U^{k_2}  
(K^{\a^2}L^{\b^2} U)&\times \cdots \\
\cdots\times U^{k_N} (K^{\a^N}L^{\b^N} U) U^{k_{N+1}}.
\end{align*}
Since $K^{\a^j} L^{\b^j} U$ is bounded from $H^s$ to $H^{-\d+|\b^j|}$, we have, 
using Lemma~A.2, 
\[
\bignorm{K^\a L^\b (U^k)}_{L^2\to H^{-N_0\d +|\b|}}
\leq C \jap{k}^{N_1},
\]
where $N_0=|\a+\b|$, $N_1= (N_0\d+|\b|)/(1-\d) + N_0$. 
Thus we learn 
\[
K^\a L^\b (f(U))\in \mathcal{B}(L^2, H^{-|\a+\b|\d+|\b|}), 
\]
and we have the following lemma: We write
\[
S^{+0}_{1,0} =\bigcap_{\d>0} S^{\d}_{1,0}.
\]

\begin{lem}\label{lem-funccal-3}
Suppose $U=\Op(u)$ is unitary with $u\in S^{+0}_{1,0}$. 
Then $f(U)$ is a pseudodifferential operator with the symbol in $S^{+0}_{1,0}$. 
\end{lem}

We then compute the principal symbol of $f(U)$. If $U=\Op(u)$ is unitary with $u\in S^\d_{1,0}$, then 
the symbol of $1=U^*U$ is $1=|u(x,\x)|^2$ modulo $S^{\d-1}_{1,0}$. 
Thus we may assume $u_0$, the principal symbol of $U$ modulo $S^{\d-1}_{1,0}$, 
has modulus 1. This implies, in particular, $u_0^j\in S_{1,\d}^0$ for any $j\geq 0$. 
We show $f(U)$ has the principal symbol $f\circ u_0$. 
We note
\begin{align*}
U^k -\Op(u_0^k) &= \sum_{j=0}^{k-1} (U^{j+1} \Op(u_0^{k-j-1}) -U^j \Op(u_0^{k-j}))\\
&=\sum_{j=0}^{k-1} U^j (U-\Op(u_0))\Op(u_0^{k-j-1}) \\
&\quad -\sum_{j=0}^{k-1} U^j (\Op(u_0^{k-j}-u_0\#(u_0^{k-j-1}))), 
\end{align*}
where $a\#b$ denotes the operator composition: $\Op(a\# b)=\Op(a)\Op(b)$. 
By the symbol calculus, we learn $u_0^{k-j} -u_0\# (u_0^{k-j-1})\in S^{\d-1}_{1,\d}$, 
and each seminorm of it is bounded by $C\jap{k}^M$ with some $M>0$. Thus, after direct 
computations, we learn that $U^k-\Op(u_0^k) \in S^{\d-1}_{1,\d}$ and its 
seminorm is bounded by $C\jap{k}^M$ with some $M$. 
Hence we have the following claim: We note $\bigcap_{\d>0}S^{\d-1}_{1,0}=\bigcap_{\d>0} S_{1,\d}^{\d-1}$. 

\begin{thm}\label{thm-functional-calculus}
Suppose $U=\Op(u)$ is unitary with $u\in S^{+0}_{1,0}$, and let $u_0$ be a principal 
symbol such that $|u_0(x,\x)|=1$. Let $f\in C^\infty(S^1)$. Then $f(U)$ is a pseudodifferential 
operator with its symbol in $S^{+0}_{1,0}$ and the principal symbol is given 
by $f\circ u_0$ modulo $S^{\d-1}_{1,0}$ with any $\d>0$. 
\end{thm}

\begin{rem}
We can actually compute the asymptotic expansion of $f(U)$ in terms of derivatives of 
$f\circ u$ and derivatives of $u$. Thus, in particular, the support of these terms are contained 
in the support of $f\circ u$, and hence the essential support of the symbol of $f(U)$ is 
contained in the support of $f\circ u$. 
\end{rem}

\begin{rem}
In our application, we consider the cace $u\in S(1,\tilde g)$, i.e., for any $\a,\b\in\ze_+d$, 
\[
\bigabs{\pa_x^\a\pa_\x^\b u(x,\x)}\leq C_{\a\b} \jap{\x}^{-|\b|}\jap{\log\jap{\x}}^{|\a|}.
\]
Then we can apply Theorem~\ref{thm-functional-calculus} to learn $f(U)$ is a 
pseudodifferential operator with the symbol 
in $S^{+0}_{1,0}$. Moreover, since the principal symbol is $f\circ u\in S(1,\tilde g)$, 
and the remainder is in $S^{-1+\d}_{1,0}$ for any $\d>0$, 
we actually learn the symbol is in $S(1,\tilde g)$. 
\end{rem}

\section{Logarithm of unitary pseudodifferential operators}
\label{appendix-log}

For notational convenience, we write
$\ell(\x) =\jap{\log\jap{\x}}$ for  $\x\in\re^d$. 
We use the following metrics on $T^*\re^d$:
\[
g= dx^2 +\frac{d\x^2}{\jap{\x}^2}, \quad \tilde g =\ell(\x)^2dx^2 +\frac{d\x^2}{\jap{\x}^2}.
\]
We recall, $a\in S(m,g)$ if and only if, for any $\a,\b\in\ze^d$, $\exists C_{\a\b}>0$ such that 
\[
\bigabs{\pa_x^\a\pa_\x^\b a(x,\x)} \leq C_{\a\b}m(x,\x) \jap{\x}^{-|\b|}, 
\quad x,\x\in\re^d, 
\]
and $a\in S(m,\tilde g)$ if and only if, for any $a\,\b\in\ze^d$, $\exists C_{\a\b}>0$ such that 
\[
\bigabs{\pa_x^\a\pa_\x^\b a(x,\x)} \leq C_{\a\b}m(x,\x) \ell(\x)^{|\a|}\jap{\x}^{-|\b|}, 
\quad x,\x\in\re^d. 
\]

\begin{ass}\label{ass-log-unitary}
Let $\g_0\in S(\ell(\x),g)$, real-valued, and $\pa_\x\g_0\in S(\jap{\x}^{-1},g)$. 
Let $U$ be a unitary pseudodifferential operator on $L^2(\re^d)$ such that the 
principal symbol is given by $e^{i\g_0}$, i.e., $U\in\Op S(1,\tilde g)$ and 
$U-\Op(e^{i\g_0}) \in \Op S(\ell(\x)/\jap{\x},\tilde g)$. 
\end{ass}

We note $e^{i\g_0}\in S(1,\tilde g)$, and natural remainder terms are in the symbol class 
$S(\ell(\x)/\jap{\x},\tilde g)$. 

\begin{thm}\label{thm-log-unitary}
Suppose $\g_0$ and $U$ as in Assumption~\ref{ass-log-unitary}. Then there is $\g\in S(\ell(\x),g)$ 
such that $U-\exp(i\Op(\g))\in \Op S(\jap{\x}^{-\infty},g)$, and $\g-\g_0\in S(\ell(\x)/\jap{\x},\tilde g)$. 
\end{thm}

\begin{lem}\label{lem-log-1}
Let $\f\in S(\ell(\x),g)$, real-valued, and $\pa_\x\f\in S(\jap{\x}^{-1},g)$.  
Then $\Op(\f)$ is essentially self-adjoint and 
$\exp(it\Op(\f)) \in \Op S(1,\tilde g)$, $t\in\re$. Moreover, 
\[
e^{it\Op(\f)}- \Op(e^{it\f}) \in \Op S(\ell(\x)/\jap{\x},\tilde g),
\]
and is uniformly bounded for $t\in [0,1]$. 
\end{lem}

\begin{proof}
The essential self-adjointness of $\Op(\f)$ follows by the commutator theorem with an auxiliary operator
$N=\jap{D_x}$. 

In order to show $e^{it\Op(\f)}\in \Op S(1,\tilde g)$, we use Beal's characterization. 
Let $K_j$ and $L_j$ ($j=1,\dots,d$) as in Appendix~\ref{appendix-funccal}. 
We note, by a simple commutator argument as in Appendix~\ref{appendix-funccal}, we can show, for any $k,\ell\in\ze$, $T>0$, 
\[
\sup_{|t|\leq T} \bignorm{\jap{D_x}^k\ell(D_x)^\ell e^{it\Op(\f)} \ell(D_x)^{-\ell}
\jap{D_x}^{-k}} _{L^2\to L^2} <\infty.
\]
We compute, for example, 
\[
L_j[e^{it\Op(\f)}] 
=i\int_0^t e^{is\Op(\f)} L_j[\Op(\f)] e^{i(t-s)\Op(\f)} ds. 
\]
Since $L_j[\Op(\f)]=\Op(\pa_{\x_j}\f) \in \Op S(\jap{\x}^{-1},g)$, we learn 
$\jap{D_x} L_j[e^{it\Op(\f)}]$ is bounded in $H^s$ with any $s\in\re$. 
Similarly, since $K_j[\Op(\f)]=\Op(\pa_{x_j}\f)\in \Op S(\ell(\x),g)$, 
we learn $\ell(D_x)^{-1} K_j[e^{it\Op(\f)}]$ is bounded in $H^s$, $\forall s\in\re$. 
Iterating this procedure, we learn, for any $\a,\b\in\ze_+^d$,
\[
\ell(D_x)^{-|\a|} \jap{D_x}^{|\b|} (K^\a L^\b[e^{it\Op(\f)}]):\ H^s\to H^s, 
\text{ bounded}, 
\]
with any $s\in\re$. By Beal's characterization, this implies $e^{it\Op(\f)}\in \Op S(1,\tilde g)$, 
and bounded locally uniformly in $t$. 

Then we show the principal symbol of $e^{it\Op(\f)}$ is $e^{it\f}$. We have 
\begin{align*}
&e^{it\Op(\f)}-\Op(e^{it\f}) 
= \int_0^t \frac{d}{ds}\bigpare{e^{is\Op(\f)} \Op(e^{i(t-s)\f})}ds \\
&\quad =  i\int_0^t e^{is\Op(\f)}\bigpare{\Op(\f)\Op(e^{i(t-s)\f})-\Op(\f e^{i(t-s)\f})} ds\\
&\quad \in \Op S(\ell(\x)/\jap{\x},\tilde g)
\end{align*}
by the asymptotic expansion. 
\end{proof}

In particular, we have 
\[
U e^{-i\Op(\g_0)}-1 \in \Op S(\ell(\x)/\jap{\x},\tilde g), 
\]
and hence there is a real-valued symbol $\g_1\in S(\ell(\x)/\jap{\x},\tilde g)$ such that 
\[
U e^{-i\Op(\g_0)}  - \Op(e^{i\g_1})\in \Op S(\ell(\x)^2/\jap{\x}^2,\tilde g). 
\]
This implies, 
\begin{equation}\label{eq-log-1}
U e^{-i\Op(\g_0)} e^{-i\Op(\g_1)} -1 \in \Op S(\ell(\x)^2/\jap{\x}^2,\tilde g). 
\end{equation}
We use the next lemma to rewrite $e^{-i\Op(\g_0)} e^{-i\Op(\g_1)}$. 

\begin{lem}\label{lem-log-2}
Let $\f\in S(\ell(\x),g)$, real-valued, and $\pa_\x\f\in S(\jap{\x}^{-1},g)$. 
Let $\y\in S(\ell(\x)^{k}/\jap{\x}^k,\tilde g)$, real-valued, with $k\geq 1$. Then 
\[
e^{i\Op(\y)} e^{i\Op(\f)}- e^{i\Op(\f+\y)} \in \Op S(\ell(\x)^{k+1}/\jap{\x}^{k+1},\tilde g).
\]
\end{lem}

\begin{proof}
We have, for any self-adjoint operators $A$ and $B$, at least formally, 
\begin{align*}
&e^{i(A+B)} e^{-iA} e^{-iB} -1 
=\int_0^1 \frac{d}{dt} \Bigpare{e^{it(A+B)} e^{-itA} e^{-itB}} dt \\
&\quad = i \int_0^1 \Bigpare{ e^{it(A+B)} (A+B-A) e^{-itA} e^{-itB} 
- e^{-t(A+B)} e^{-itA} B e^{-itB}} dt \\
&\quad = i \int_0^1 e^{it(A+B)} \bigbrac{B, e^{-itA}} e^{-itB} dt \\
&\quad = -\int_0^1\biggpare{\int_0^t e^{it(A+B)} e^{i(t-s)A} [A,B] e^{-isA} e^{-itB} ds} dt.
\end{align*}
This computation is easily justified when $A=\Op(\f)$ and $B=\Op(\y)$, and since 
$[\Op(\f),\Op(\y)]\in \Op S(\ell(\x)^{k+1}/\jap{\x}^{k+1},\tilde g)$, 
$e^{it\Op(\f)} \in \Op S(1,\tilde g)$, etc., we have 
\[
e^{i\Op(\f+\y)}e^{-i\Op(\f)}e^{-i\Op(\y)} -1 \in \Op S(\ell(\x)^{k+1}/\jap{\x}^{k+1},\tilde g), 
\]
and this implies the assertion. 
\end{proof}

\begin{proof}[Proof of Theorem~\ref{thm-log-unitary}]
Combining \eqref{eq-log-1} with lemma~\ref{lem-log-2}, we have 
\[
U e^{-i\Op(\g_0+\g_1)} -1 \in \Op S(\ell(\x)^2/\jap{\x}^2,\tilde g). 
\]
We note $\g_0+\g_1\in S(\ell(\x),g)+S(\ell(\x)^2/\jap{\x},\tilde g) \subset S(1,g)$. 
Iterating this procedure, we construct $\g_k\in S(\ell(\x)^{k}/\jap{\x}^k,\tilde g)$, 
real-valued, such that 
\[
U e^{-i\Op(\g_0+\cdots+\g_k)} -1 \in \Op S(\ell(\x)^{k+1}/\jap{\x}^{k+1},\tilde g). 
\]
for $k=2,3,\dots$. Then we choose an asymptotic sum: $\g\sim \sum_{k=0}^\infty \g_k$, 
i.e., $\g\in S(\ell(\x),g)$ and 
\[
\g-\sum_{k=0}^N\g_k \in S(\ell(\x)^{N+1}/\jap{\x}^{N+1},\tilde g)
\]
for any $N>0$. Then we have 
\[
U e^{-i\Op(\g)}-1\in \Op S(\jap{\x}^{-\infty},\tilde g)=\Op S(\jap{\x}^{-\infty},g),
\]
and we complete the proof of Theorem~\ref{thm-log-unitary}. 
\end{proof}

\section{Trace class scattering for unitary operators}
\label{appendix-KB}

The next theorem, the unitary version of the Kuroda-Birman theorem, 
seems well-known, but the author could not find an appropriate 
reference. Here we give a proof for the completeness. 

\begin{thm}\label{thm-trace-scat}
Let $U_1$ and $U_2$ be unitary operators on a separable Hilbert space, 
and suppose $U_1-U_2$ is a trace class operator. Then 
$\s_{\mathrm{ac}}(U_1)=\s_{\mathrm{ac}}(U_2)$.  
\end{thm}

\begin{proof}
Since the eigenvalues of $U_1$ and $U_2$ are at most countable, we can find $\th\in\re$
such that $e^{-i\th}$ is not an eigenvalue of both $U_1$ and $U_2$. Then, by replacing 
$U_1$ and $U_2$ by $e^{i\th}U_1$ and $e^{i\th}U_2$, respectively, we may suppose 
1 is not an eigenvalue of both $U_1$ and $U_2$. Then we can define the Cayley transform of $U_1$
and $U_2$ by 
\[
H_j =i(U_j+1)(U_j-1)^{-1}, \quad j=1,2. 
\]
By the definition, we have
\[
U_j =(H_j+i)(H_j-i)^{-1}= 1+2i(H_j-i)^{-1}, \quad j=1,2, 
\]
and hence 
\[
(H_1+i)^{-1}-(H_2+i)^{-1} = \frac{1}{2i}(U_1-U_2),
\]
is in the trace class. Thus we can apply the Kuroda-Birman theorem (\cite{Reed-Simon}, Theorem~XI.9)
to learn $\s_{\mathrm{ac}}(H_1)= \s_{\mathrm{ac}}(H_2)$. This implies the assertion since
\[
\s_{\mathrm{ac}}(U_j)= \bigset{(s-i)(s+i)^{-1}}{s\in\s_{\mathrm{ac}}(H_j)}, 
\quad j=1,2,
\]
by the spectral decomposition theorem. 
\end{proof}


\end{document}